\documentclass[aps,pre,twocolumn,showpacs]{revtex4-1}
\usepackage{amssymb}
\usepackage{graphicx}
\usepackage{dcolumn}
\usepackage{bm}
\usepackage{amsmath}
\usepackage{soul,color}
\usepackage{textcomp}
\usepackage{times}
\usepackage{enumitem}

\definecolor{cGreen}{RGB}{0,153,76}
\definecolor{cBlue}{RGB}{45,51,180}
\definecolor{cmagenta}{RGB}{205,0,100}

\usepackage[colorlinks,linkcolor=cBlue,citecolor=cBlue]{hyperref}

\newtheorem{prop}{Proposition}

\newenvironment{proof}{{\noindent\it Proof}\quad}{\hfill $\square$ \par}

\begin{document}

\preprint{APS/123-QED}

\title{Complex-Valued Neural Operator Assisted Soliton Identification}
\author{Ming Zhang$^{1}$}
\author{Qi Meng$^{2}$}
\email{meq@microsoft.com}
\author{Deng Zhang$^{3}$}
\author{Yue Wang$^{2}$}
\author{Guanghui Wang$^{1}$}
\author{Zhiming Ma$^{4}$}
\author{Li Chen$^{5}$}
\email{lchen@sxu.edu.cn}
\author{Tie-Yan Liu$^{2}$}
\email{tyliu@microsoft.com}
\affiliation{$^1$School of Mathematics, Shandong University, Jinan, 250100, China.\\
$^2$Microsoft Research, Beijing, 100080, China.\\
$^3$School of Mathematical Sciences, Shanghai Jiao Tong University, Shanghai, 200240, China.\\
$^4$Academy of Mathematics and System Sciences, Chinese Academy of Sciences, Beijing, 100190, China\\
$^5$Institute of Theoretical Physics, State Key Laboratory of Quantum Optics and Quantum Optics Devices, Shanxi University, Taiyuan 030006, China.}

\begin{abstract}
The numerical determination of solitary states is an important topic for such research areas as Bose-Einstein condensates, nonlinear optics, plasma physics, etc. In this paper, we propose a data-driven approach for identifying solitons based on dynamical solutions of real-time differential equations. Our approach combines a machine-learning architecture called the complex-valued neural operator (CNO) with an energy-restricted gradient optimization. The former serves as a generalization of the traditional neural operator to the complex domain, and constructs a smooth mapping between the initial and final states; the latter facilitates the search for solitons by constraining the energy space. We concretely demonstrate this approach on the quasi-one-dimensional Bose-Einstein condensate with homogeneous and inhomogeneous nonlinearities. Our work offers a new idea for data-driven effective modeling and studies of solitary waves in nonlinear physical systems.
\end{abstract}

\maketitle

\section{Introduction}
{Steady soliton solutions of nonlinear partial differential equations (PDEs) arise in a wide range of contexts in physics,} including Bose-Einstein condensates (BECs) \cite{Frantzeskakis2010, Strecker2003, Kartashov2011}, nonlinear optics \cite{Kartashov2011, Kivshar1993, Maimistov2010, Dudley2009}, and plasma physics \cite{Kauranen2012, Kuznetsov1986}. {Since analytical solutions to nonlinear PDEs are generally difficult to obtain (especially for non-integrable PDEs), numerical identification of solitons constitutes an important subject for both theoretical inquiries and practice. From a theoretical perspective, the ground-state solitary solution can unveil the equilibrium characteristics of nonlinear systems, such as long-range order and topological structures; the lifetime of soliton solutions can offer valuable insights into the stability and response to perturbations near equilibration. From the practical point of view, the study of solitons has potential implications in long-distance communication \cite{Ablowitz2000, Haus1996, Hasegawa2000}, data transmission, and in laying the foundation for advanced photonic devices and data storage technologies \cite{Wu2004, Hong2003, Xu2017}.}

{The solitary waves that we are interested in are separable in time and space. Several traditional methods are available for finding this type of solitons \cite{Yang2010, Kelley2003, Boyd2001}. For example, the complex evolution method extends the real-time of PDEs to complex or imaginary values \cite{Bao2004, Yang2008}. Since complex-time evolution always reduces energy, this method is more suitable for finding ground-state solitons. Moreover, several methods have been developed based on stationary equations for both ground-state and excited-state solitons, where the time dependence has been eliminated using the space-time separation. For instance, the shooting method is primarily used for 1D systems; the Petviashvili method \cite{Petviashvili1976} and its extensions \cite{Musslimani2004, Ablowitz2005, Lakoba2007} are suitable for finding higher-dimensional ground state solitons; Newton's method \cite{Kelley2003, Boyd2001} and its advanced versions, such as the CG Newton's method \cite{Yang2009}, iteratively search for ground state or excited state solitons starting from trial solutions and can be extended to higher dimensional systems.} Some machine learning methods have also been applied in finding the solitons, e.g., the variational neural network ansatz \cite{Luo2022}, the deep residual \cite{Nabian2019}, convolutional neural network \cite{Karumuri2020, Gao2021}, as well as the generative models \cite{Zhu2019}. Furthermore, there are also some dynamical PDE solvers based on machine learning \cite{Lagaris1998, Blechschmidt2021}. {Physics-informed neural networks have been proposed as a powerful tool to approximate the dynamical solutions by incorporating the governing equations as soft constraints during the training process \cite{Raissi2019, Raissi2018}. The Feynman-Kac formula-based methods \cite{Beck2019, Beck2021} and stochastic equation-based methods \cite{Blechschmidt2021, Beck2020} have been reported as well.} This class of solvers, however, is designed to find the dynamical solutions of PDEs for given initial states, which cannot be used for determining solitary solutions directly.

In this paper, we propose {a data-driven approach to search for the solitons based on a machine-learning architecture called the complex-valued neural operator (CNO). We were motivated by the question of that is it possible to identify solitary wave solutions by directly looking at the real-time PDE, rather than its variants such as the imaginary-time PDE or the stationary equation. Since solitary waves exhibit space-time separation, in principle, we can use traditional PDE solvers (e.g., Euler or Runge-Kutta (RK) solvers) to identify solutions where the initial and final states differ by only a phase factor. There are two challenges for such a task. First, iterating through the initial states by PDE solvers is highly time-consuming. Second, finding a specific solution is an optimization problem, and hence it would be ideal to use the gradient descent. Traditional PDE solvers, however, cannot compute the derivative with respect to initial states. These two challenges can be well addressed by the neural operator (NO) \cite{li2020fourier, Kovachki2021}, which can establish the continuous mapping between the real-value functions during training. We extend the NO to the complex domain, namely the CNO with complex layers and complex activation functions, to accommodate the complex dynamical PDEs. We further develop an energy-restricted optimization algorithm to reduce the spaces of states during the search process.} Our approach is concretely demonstrated on the one-dimensional Gross-Pitaevskii (GP) equations with both homogeneous and inhomogeneous nonlinearities. Furthermore, we show that the trained CNO can be also used for the stability analysis of solitary states.

The rest of this paper is organized as follows: In Sec.~\ref{sec:method}, we present the basic idea of our approach, { and show the CNO architecture.} In Sec.~\ref{sec:Soliton Search}, we use our method in both the homogeneous GP and the inhomogeneous GP equations to identify the solitons. In Sec.~\ref{sec:NoiseAnalysis}, we present the application of the CNO to the stability analysis of solitons. A brief summary and outlook can be found in Sec.~\ref{sec:Summary}. 

\section{General Method} \label{sec:method}
Let us generically consider a class of nonlinear PDEs of a one-dimensional physical system
\begin{equation}
    f(\psi,\dot{\psi},t)=0,
    \label{GPDE}
\end{equation}
with $\psi(x,t)$ being a complex-valued function. From the perspective of field theory, the PDEs come from the Euler-Lagrange equation $\partial_{\psi} L = d_t (\partial_{\dot{\psi}}L)$ with $L(\psi,\dot{\psi},t)$ being the Lagrangian. For example, the GP equation is generated by the nonlinear Schr{\"o}dinger Lagrangian, and the nonlinear Klein-Gordon equation arises from the Klein-Gordon Lagrangian with mass or high-order potentials. The solitons that we are interested in are space-time separable, i.e.,
\begin{equation}
\psi(x,t)=e^{-i\alpha t}\phi(x),
    \label{psi}
\end{equation}
where $i = \sqrt{-1}$ is the imaginary unit, and $e^{-i\alpha t}$ is a time-dependent phase factor that is isolated from the spatial solitary state $\phi(x)$. Using the time-space separation of $\psi(x,t)$, the temporal degree-of-freedom of Eq.~(\ref{GPDE}) can be eliminated, which leads to the stationary equation purely in terms of $\phi(x)$, solving which one can obtain the time-independent solitons. {Due to the nonlinearity of $f$, the stationary equation is generally not a linear eigenstate equation, and various iterative methods have been developed for such stationary equations \cite{Kelley2003, Yang2010}.}

\subsection{Basic Idea}
{We aim to identify solitary solutions by directly looking at the real-time dynamical PDEs.} In such a context of {real-time dynamics, the PDE [Eq.~(\ref{GPDE})]} maps the initial states to the final states, i.e., $G: \psi(x,0) \mapsto \psi(x,T)$, and the solitary solutions correspond to the fixed points of the mapping with unit fidelity $F=1$, where
\begin{equation}
    F=\frac{1}{N^2}\left|\int dx \psi^*(x,T) \psi(x,0)\right|^2,
    \label{F}
\end{equation}
and $N = \int dx |\phi(x)|^2$ is the total particle number, and $\psi^*(x,T)$ is the complex conjugation of $\psi(x,T)$.
Considering the fact that the space of states is continuous and quite large, it's not feasible to locate the solitons by directly iterating through the initial states $\psi(x,0)$ using PDE solvers. {To address this issue, we developed the CNO and a restricted searching algorithm. Our algorithm majorly consists of two steps: i) Learn the initial-final state mapping $G$ by training the CNO based on a finite-size dataset $\mathcal{D}$; ii) Perform restricted gradient optimization to find the solitons within a specific energy range. Below we explain in detail.} 
\begin{figure}[t]
	\includegraphics[width=0.48\textwidth]{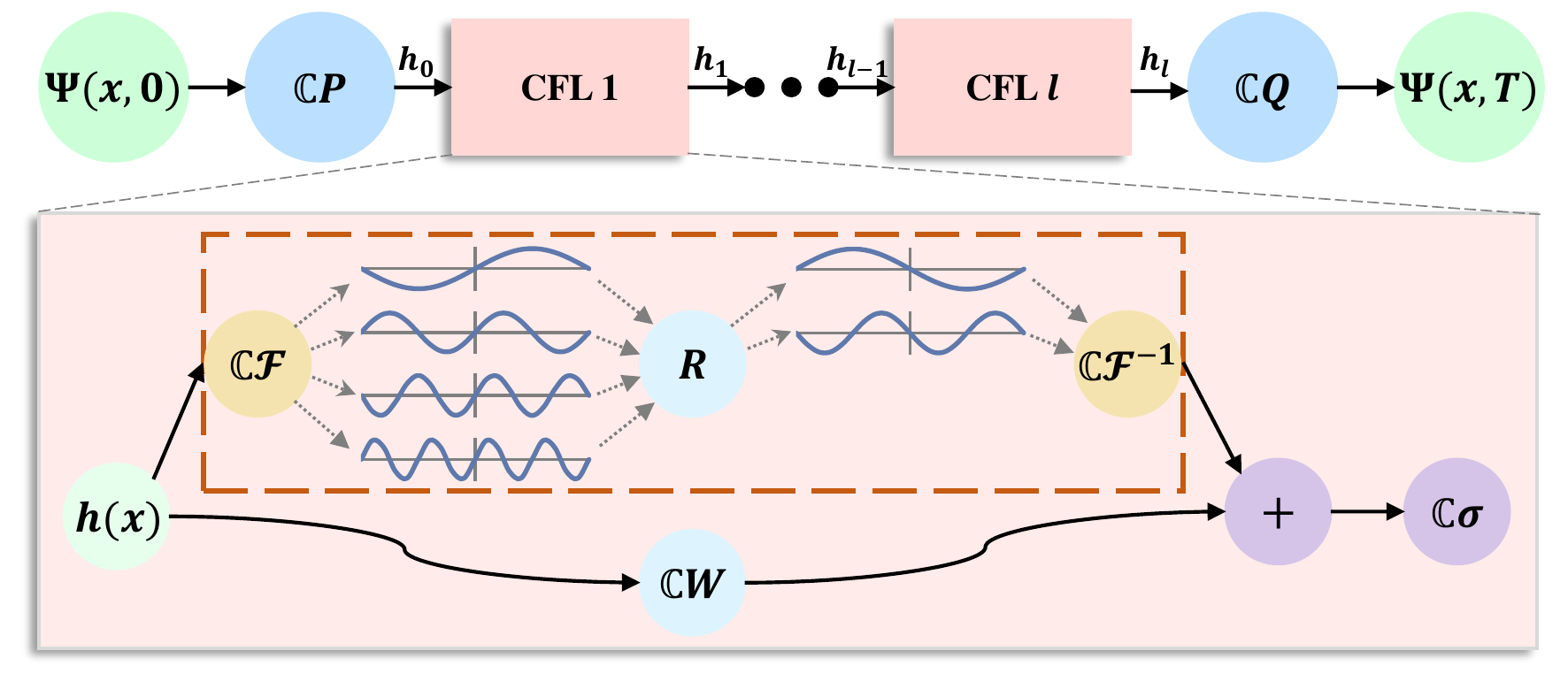}%
	\caption{Upper panel: Architecture of a CNO where initial $\psi(x,0)$ and final $\psi(x,T)$ states serve as the input and output, {$\mathbb{C}P$ and $\mathbb{C}Q$} correspond to the complex embedding and decoding layers, and $h(x)$ are the hidden states. Lower panel: Detailed structure of a CFL with {$\mathbb{C}\mathcal{F}$ and $\mathbb{C}\mathcal{F}^{-1}$} respectively the complex Fourier and inverse Fourier transformations, $R$ the Low-pass filter layer, {$\mathbb{C}W$ a local complex transformation, and $\mathbb{C}\sigma$ the complex-valued nonlinear activation}.}
	\label{fig:CNO}
\end{figure}

{As is schematically shown in Fig.~\ref{fig:CNO}, CNO is a trainable machine-learning architecture that is able to learn the relation between the initial state $\psi(x,0)$ (as the input) and the final state $\psi(x,T)$ (as the output) based on a finite-size dataset $\mathcal{D}$. More detailed descriptions of CNO can be found in subsection~\ref{CNO}. The dataset $\mathcal{D}$ is composed of the pairs of initial-final states, i.e.,
\begin{equation}
    \mathcal{D}=\{(\psi_k(x,0),\psi_k(x,T))|k=1,2,...,M\},
    \label{D}
\end{equation}
with $\psi_k(x,0)$ being random differentiable functions, and $M$ being the total number of data points. Practically, $\psi_k(x,T)$ can be generated by the numerical propagating Eq.~(\ref{GPDE}) by certain PDE solvers, e.g., the Euler or the RK solvers.} CNO learns the mapping $G$ by minimizing the loss function $\mathcal{L}$, which refers to the training process. The loss function $\mathcal{L}$ is defined to characterize the effective distance between the network prediction $\psi_\text{pred}(x,T)$ and the solutions $\psi(x,T)$ in the dataset $\mathcal{D}$. Practically, in the following discussion we use the $L_2$-norm loss function,
\begin{equation}
    \mathcal{L}=\frac{1}{M}\sum_k\int dx\frac{\left|\psi_{k,\text{pred}}(x,T)-\psi_k(x,T)\right|}{\left|\psi_k(x,T)\right|}.
    \label{equ:l2loss}
\end{equation}

Once the mapping $G$ is learned by the CNO, the trained CNO can be used to find the solitons by minimizing $\Delta F=|F-1|$ with respect to the initial state $\psi(x,0)$. {This is an optimization process, which can be carried out using the gradient descent method. Here, we emphasize that the CNO-based soliton search has two major advantages. First, the forward propagation of CNO is much faster than traditional PDE solvers. The detailed discussion on the computational complexity will be shown in subsection~\ref{CNO}. In Sec.~\ref{sec:NoiseAnalysis}, we additionally show a concrete running-time comparison between CNO and traditional PDE solvers in the context of stability analysis of solitons. Furthermore, the output of the CNO is continuously derivable to its input, which thus naturally facilitates the gradient-based optimization of $\Delta F$.

Although the trained CNO can provide several computational conveniences as mentioned above, the vast space of initial states remains a significant obstacle in our practical pursuit of solitons. This motivated us to reduce the search space, which is based on Noether's theorem. Noether's theorem establishes a relationship between symmetries and conserved charges, with the latter being of help in space reduction.} For example, in systems with time translation symmetry, such as the GP equation or nonlinear Klein-Gordon equation, the energy (or the Hamiltonian) $E=\int dx \Pi\dot{\psi} -L$ is conserved, with $\Pi=\partial_{\dot{\psi}}L$ being the conjugate field. {Therefore, we can perform the energy-restricted optimization by adding an inequality constraint into the minimization of $\Delta F$.} To be more specific, if we set an energy upper bound $E\le E_\text{max}$, the optimization problem is reformed as 
\begin{equation}
\begin{aligned}
    \min_{\psi(x,0)} \Delta F \text{\ \ \ \ \ subject to \ \ \ } E_{\psi(x,0)} \le E_\text{max}.
    \label{equ:OP}
\end{aligned}
\end{equation}
Using the augmented Lagrangian method \cite{bertsekas2014constrained}, the above optimization problem can be rewritten as
\begin{equation}
\begin{aligned}
   \min_{\psi(x,0),\lambda, \gamma} \Delta F + \frac{\gamma}{2}\left(\max\{\frac{\lambda}{\gamma} + E_{\psi(x,0)}- E_\text{max},0\}^2-\frac{\lambda^2}{\gamma^2}\right),
    \label{equ:OP_New}
\end{aligned}
\end{equation}
where $E_{\psi(x,0)}$ is the energy of the input $\psi(x,0)$, $\lambda$ is the Lagrange multiplier, and $\gamma$ is the penalty factor. 

{
\subsection{Complex-Valued Neural Operator}
\label{CNO}
Here, we in detail discuss the CNO.} Neural operators (NO) are machine learning models proposed to learn mappings between functions \cite{li2020fourier, Kovachki2021}. {Since the mapping $G$ here is generally complex, we have to extend the NO from the real to the complex domain.} Taking the Fourier neural operator as a backbone \cite{li2020fourier}, we propose the {Complex-valued NO, i.e., the CNO. There have been some recent studies on the complex-valued neural networks \cite{Scardapane2018, Bassey2021, Trabelsi2017Deep, arjovsky2016unitary, guberman2016complex}. One known challenge lies in the fact that complex-valued activation functions are not simultaneously complex-differentiable and bounded \cite{Scardapane2018, Bassey2021}, which leads to complex neural networks still being an open and active research topic.}

Our CNO, as shown in Fig.~\ref{fig:CNO}, is formulated as a multi-layer architecture with the initial state $\psi(x,0)$ being the input and the final state $\psi(x,T)$ the output, i.e.,  {
\begin{equation}
\psi(x,T)=\mathbb{C}Q(\text{CFL}^l(\cdots(\text{CFL}^1(\mathbb{C}P(\psi(x,0)))))).
\end{equation}
Here, $\mathbb{C}{P}(\psi(x,0))$ and $\mathbb{C}Q(h_l)$ layers are the complex embedding and complex decoding layers,} where the former can lift the input to a higher-dimensional hidden space to ensure the expressiveness of the model, while the latter works oppositely to make the output have the same dimension as the input.

Feature learning takes place in the hidden space, which consists of several complex-valued Fourier layers (CFLs). In CFL $j+1\ (j=0,1,...,l-1)$, the hidden state $h_j$ are projected to $h_{j+1}$ as{
\begin{equation}
h_{j+1}(x) = \mathbb{C}\sigma\left[\mathbb{C}Wh_j(x)+\mathbb{C}\mathcal{F}^{-1}(R\cdot \mathbb{C}\mathcal{F}(h_j(x)))\right],
\label{equ:CNO}
\end{equation}
where $\sigma$ denotes the complex-valued element-wise non-linear activation, $\mathbb{C}W$ is a complex-valued convolutional network which implements the linear transformation on $h_j$, $\mathbb{C}\mathcal{F}$ and $\mathbb{C}\mathcal{F}^{-1}$ are the complex Fourier transformation and inverse Fourier transformation,} and $R$ denotes the Low-pass filter defined on the frequency space. {The main points in the construction of CFLs lie in the complex-valued convolution $\mathbb{C}Wh_j(x)$ and the complex-valued activation $\mathbb{C}\sigma(\cdot)$. Below, we will discuss them one by one:}
\begin{itemize}[leftmargin=12pt]
\item {Complex convolution. To perform the equivalent operation of traditional real-valued 2D convolution in the complex domain, we convolve the hidden complex vector $h=a+ib$ with the complex filter matrix $\mathbb{C}W = A + iB$, where $A$ and $B$ are real matrices, and $a$ and $b$ are real vectors.} Since the convolution operator is distributive, convolving the vector $h$ by the filter $\mathbb{C}W$ can be simply expressed by
\begin{equation}
\mathbb{C}W*h = (A*a-B*b)+i(B*a+A*b),
\label{equ:C_conv}
\end{equation}
{with $*$ denoting the convolutional operation.}

\item Complex-value activation. {We generalize the real-valued GELU (Gaussian error linear unit) \cite{Hendrycks2016} to its complex counterpart, namely $\mathbb{C}$GELU. }
The $\mathbb{C}$GELU is defined as 
\begin{equation}
\mathbb{C}\text{GELU}(z)=\text{GELU}(\mathcal{R}(z))+i\text{GELU}(\mathcal{I}(z)),
\label{equ:CReLU}
\end{equation}
where both the real $\mathcal{R}(z)$ and imaginary $\mathcal{I}(z)$ parts of a neuron are activated by GELU, {with 
\begin{equation}
\text{GELU}(x) = x\cdot \frac{1}{2}[1+\text{erf}(x/\sqrt{2})],
\end{equation}
$\text{erf}(\cdot)$ the Gauss error function, and $x$ being a real number. 
It is known that GELU is non-convex and non-monotonic, and has been practically applied in many large language models (e.g., OpenAI's GPT \cite{Radford2018} and Google AI's BERT models \cite{Devlin2018}) where it outperforms the convex and monotonic ReLU. Note that we may have some other extensions, such as activating the norm of $z$ while keeping its phase factor unchanged, i.e., $\mathbb{C}\text{GELU}(x) = \text{GELU}(|z|+z_0) \exp(i \arg{z})$ with $z_0$ being a real learnable parameter. The performance of various extensions remains to be further studied, but this is not the focus of this article. Therefore, in the following, we consistently adopt the $\mathbb{C}\text{GELU}$ defined in Eq.~(\ref{equ:CReLU}).}
\end{itemize}

{Now, we discuss the complexity of CNO. In a CFL, the time complexity majorly comes from the Fourier and the inverse Fourier transforms which provides a complexity $ O(n \log{n})$ with $n$ being the dimension of the input. Hence, the total complexity of a CNO with $l$ CFLs is $O(l n \log{n})$.} On the other hand, traditional explicit PDE solvers based on the finite difference and the pseudo-spectrum are well-known to exhibit complexity $O(t_s n^2)$ and $O(t_s n \log{n})$, respectively, with $t_s$ being the time steps. Hence, the forward propagation of CNO should be significantly faster than traditional explicit solvers as $l\ll t_s$. { This property not only facilitates our soliton identification algorithm mentioned above, but also speeds up the stability analysis that will be shown in Sec.~\ref{sec:NoiseAnalysis}. We will numerically compare the propagation efficiency of the CNO and traditional PDE solvers in Sec.~\ref{sec:NoiseAnalysis}.}

\section{Soliton Identification}
\label{sec:Soliton Search}
We demonstrate our approach on the one-dimensional GP equation 
\begin{equation}
    i \partial_t \psi = \left[-\frac{\partial_x^2}{2m}+V(x)+g(x)\left|\psi\right|^2\right] \psi,
    \label{1DGPE}
\end{equation}
which is generated by the Lagrangian (setting $\hbar=1$)
\begin{equation}
    L=\int dx \left[i \psi^*\dot{\psi}-\frac{|\partial_x\psi|^2}{2m}-V(x)|\psi|^2-\frac{g(x)}{2}|\psi|^4\right].
    \label{GPELagrangian}
\end{equation}
It describes a quasi-one-dimensional BEC with $N$ atoms being tightly confined in the transverse $y$-$z$ directions, with $\psi(x,t)$ being the longitudinal field satisfying $\int dx |\psi(x,t)|^2 = N$.
The first, second, and third terms on the R.H.S. of Eq.~(\ref{1DGPE}), respectively, denote the kinetic energy, the longitudinal potential along $x$, and the nonlinear interaction arising from the two-body s-wave collision of atoms with $g=2 a_s/m\ell_\perp^2$ the reduced nonlinearity strength, $a_s$ the scattering length, $\ell_\perp$ the transverse confinement length, and $m$ the atomic mass. $g$ is commonly set positive to stabilize the BEC \cite{Pethick2008,Pitaevskii2003}. Particularly, 
an equation with a space-independent interaction $g(x) = g_0$ is known as the homogeneous GP equation, whereas if $g(x)$ is space-dependent, the equation is called the inhomogeneous GP equation \cite{Beitia2007, Sivan2006}. Practically, spatial inhomogeneity $g(x)$ can be achieved using the confinement-induced resonance technique \cite{Olshanii1998, Bergeman2003}, i.e., properly engineering the transverse confinement near the orbital resonance point to induce spatial inhomogeneity in the scattering length. The energy of the BEC is simply the Hamiltonian of the Schr{\"o}dinger field, i.e., 
\begin{equation}
    E = \int dx \left[\frac{|\partial_x\psi|^2}{2m}+V(x)|\psi|^2+\frac{g(x)}{2}|\psi|^4\right],
    \label{E}
\end{equation}
which is conserved during the time evolution, as mentioned before. 

{It is known that almost all solitary solutions of the 1D GP equation can be anticipated using Newton's method or other eigensolvers, based on trial solutions constructed from the eigenstates of the linear Schrödinger equation (with the nonlinearity turned off) or their variations. These solitons provide the ground truth to verify the effectiveness of our algorithm, which is the main objective of this paper. On the other hand, the trade-off is that there are rarely unexpected solutions in the 1D system and hence generalizing our algorithm to higher-dimensional systems would be more practically valuable. This generalization is not so straightforward as will be discussed in Section \ref{sec:Summary}.}

\subsection{Homogeneous GP Equation}
For the first example, we consider a homogeneous BEC confined in a harmonic potential $V(x)=\omega^2 x^2/2$ with fixed $g(x)=g_0=20\omega/N\ell$ and $\ell=(m\omega)^{-1/2}$ denoting the harmonic length. We generate a training dataset $\mathcal{D}_\text{train}$ of size $M_\text{train}=10^4$ and a testing dataset $\mathcal{D}_\text{test}$ of size $M_\text{test}=10^3$ for such a learning task, where the data points are collected by independently propagating the GP equation Eq.~(\ref{1DGPE}) using the RK finite difference method \cite{Scherer2013} from the initial states $\psi(x,0)$ to the final state $\psi(x,T)$ with $T=\omega^{-1}$ being fixed. The initial states $\psi(x,0)$ are randomly generated using the harmonic basis, i.e., $\psi(x,0) = \sum_{n=1}^{n_c} c_n \xi_n$ where $c_n$ denotes the expansion coefficients satisfying $\sum_n |c_n|^2=1$, $n_c=20$ is the basis cutoff, and 
\begin{equation}
    \xi_n(x) = \frac{1}{2^n n!}\left(\frac{1}{\pi R^2}\right)^{1/4}H_n\left(\frac{x}{R}\right)e^{-\frac{x^2}{2R^2}}
    \label{xiHG}
\end{equation}
with $H_n$ being the Hermite polynomial and $R=(3gN\ell^2/2\omega)^{1/3}$ being the Thomas-Fermi radius. Employing the harmonic basis, rather than directly randomizing the $\psi(x,0)$ in the coordinate space is based on the physical consideration that the complex wave $\psi(x,t)$ and its first-order derivative $\partial_x\psi(x,t)$ should be continuous for a potential $V(x)$ without singularity. However, the soliton identification is still conducted in the coordinate space.

\begin{figure}[!ht]
    \includegraphics[width=0.48\textwidth]{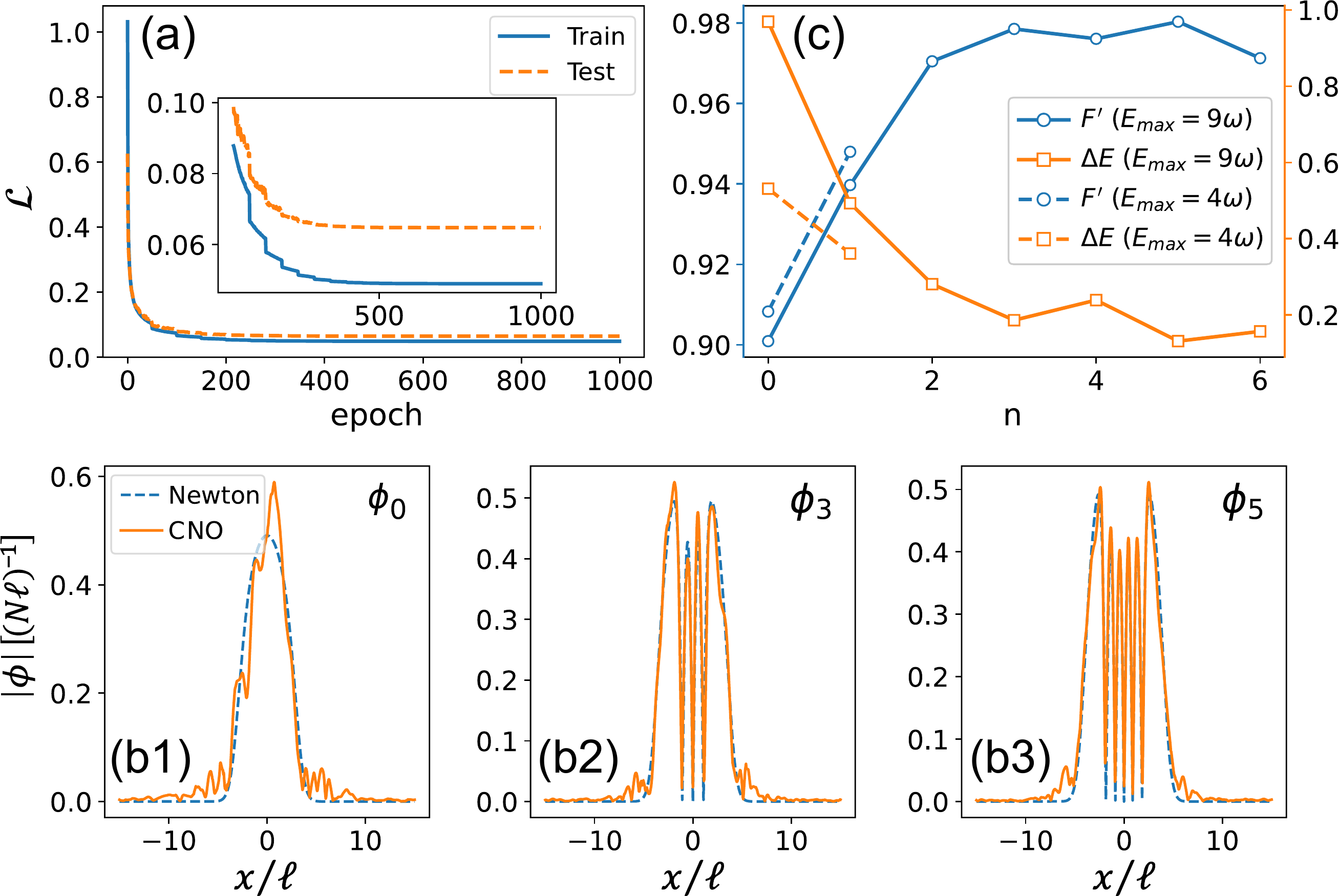}
   \caption{Numerical results of the homogeneous GP equation. (a) Variation of the loss function $\mathcal{L}$ on epochs during the training process, where the solid and dashed curves indicate $\mathcal{L}$ on the training $\mathcal{D}_\text{train}$ and testing  $\mathcal{D}_\text{test}$ datasets, respectively. (b) Comparison between the solitary states identified by the CNO (solid lines) and those obtained by Newton's method (dashed lines). Panels (b1), (b2) and (b3) display the amplitude of the ground-state $|\phi_0|$ , the 3rd excited-state $|\phi_3|$, and the 5th excited-state $|\phi_5|$, respectively. (c) The overlap $F'$ (circle) and the energy error $\Delta E$ (square) as a function of the soliton index $n$. Solid lines correspond to the case of $E_\text{max}=9\omega$, whereas dashed lines correspond to the case of $E_\text{max}=4\omega$.}
	\label{fig:x2}%
\end{figure}

Through feeding $D_\text{train}$ to the CNO, we carry out the training by minimizing the loss function $\mathcal{L}$ [Eq.~(\ref{equ:l2loss})] using the Adam optimizer \cite{Kingma2014}. Our CNO contains $l=4$ CFLs, and the embedding dimension is $64$.
Fig.~\ref{fig:x2}(a) presents $\mathcal{L}$ as a function of training epochs, where the solid and dashed curves indicate $\mathcal{L}$ on the training $\mathcal{D}_\text{train}$ and the testing $\mathcal{D}_\text{test}$ datasets, respectively. The inset takes a closed look at $\mathcal{L}$ within the range $[0.04,0.1]$. It is shown that the training process converges at about $450$ epochs, {as indicated by the loss function $\mathcal{L}$ reaching a broad plateau. The training error and testing error exhibit similar behavior with small quantitative differences, which indicates the trained CNO does not suffer from a severe overfitting problem.}

After convergence, we perform the energy-restricted optimization [Eq.~(\ref{equ:OP_New})] to identify the soliton states based on the trained CNO. Practically, we set the energy bound to $E_\text{max}=9\omega$, and optimize Eq.~(\ref{equ:OP_New}) using Adam optimizer from $10^3$ stochastic initial states.
As a result, the lowest 7 solitary states can be identified. In Fig.~\ref{fig:x2}(b1)-(b3), we typically display the amplitude of the ground-state $|\phi_0|$, the 3rd excited-state $|\phi_3|$, and the 5th excited-state $|\phi_5|$ by solid lines, respectively. As a comparison, we also show the amplitude of solitons $|\phi_n|$ obtained by Newton's method by dashed lines, which serve as the ground truth. Both results are in good qualitative agreement. {Quantitatively, the solitons found by the CNO shown in Fig.~\ref{fig:x2}(b) do not perfectly matches the ground truth. The mismatch can be attributed to the intrinsic error of CNO in learning the mapping $G$. In fact, training errors are inevitable for any data-driven machine-learning model.}

To quantitatively measure the discrepancy between the solitons found by the CNO and Newton's method, we calculate the overlap (solid line with circles) 
\begin{equation}
    F'=\left|\int dx \phi_{n,\text{CNO}}^*(x)\phi_{n,\text{Newton}}(x)\right|^2,
    \label{Fp}
\end{equation}
and energy discrepancy (solid line with squares) 
\begin{equation}
    \Delta E= \left| E_n^\text{CNO} - E_n^\text{Newton}\right|,
    \label{DE}
\end{equation}
as are plotted in Fig.~\ref{fig:x2}(c).
One can observe that the identified solitons averagely have $F'\ge 96\%$ and $\Delta E\le0.35$. In Fig.~\ref{fig:x2}(c), $\Delta E$ roughly decreases monotonically as $n$ increases, which means the solitons $\phi_{n,\text{CNO}}$ closer to the energy bound $E_\text{max}$ exhibit a lower energy error $\Delta E$. This phenomenon can be understood that $F$ [Eq.~(\ref{F})] is generally non-convex near a fixed point, i.e., there are several local minima $\Delta F\gtrsim 0$ near the point $\Delta F=0$. The energy constraint (the last term of Eq.~(\ref{equ:OP_New})) becomes more and more important as $E_\psi(x,0)$ approaches the bound $E_\text{max}$, thereby decreasing the likelihood that the algorithm will achieve a local minimum. Otherwise, if $E_\psi(x,0)$ is too far away from $E_\text{max}$, the last term of Eq.~(\ref{equ:OP_New}) is simply a constant referring to no restriction. In this case, the searching algorithm may converge to a local-minimum solution giving rise to a large $\Delta E$. A mathematical proof of this statement can be found in Appendix \ref{app:ERS}. For numerical verification, we set a small bound $E_\text{max}=4\omega$ such that only the lowest two solitons ($\phi_0$ and $\phi_1$) are allowed to be identified. The corresponding $F'$ and $\Delta E$ are plotted by dashed lines in Fig.~\ref{fig:x2}(c), from which one can clearly observe that a lower bound $E_\text{max}$ is really helpful to reduce the error $\Delta E$.

\begin{figure}[t]
    \includegraphics[width=0.48\textwidth]{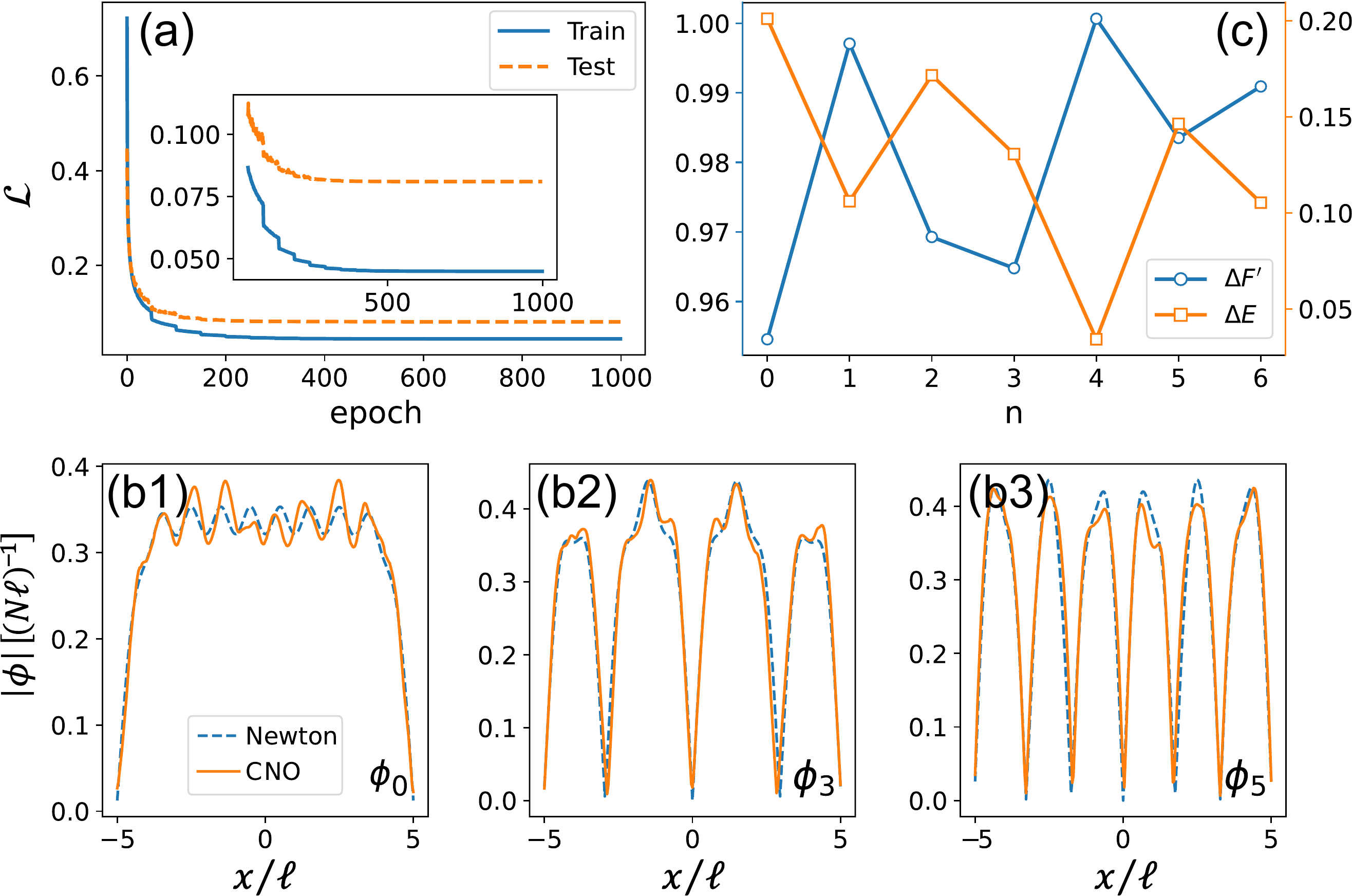}
	\caption{Numerical results of the inhomogeneous GP equation. (a) Variation of the loss function $\mathcal{L}$ on epochs during the training process. (b) Comparison between the solitary states identified by the CNO (solid lines) and those obtained by Newton's method (dashed lines). (c) The overlap $F'$ (circle) and the energy discrepancy $\Delta E$ (square) as a function of the soliton index $n$.}
	\label{fig:box}
\end{figure}

\subsection{Inhomogeneous GP Equation}
We illustrate the second example by considering a BEC carrying inhomogeneous nonlinearity $g(x)=g_0+\delta g \sin(2\pi x)$ and being confined in a boxed potential within $x\in[-R,R]$, where the homogeneous $g_0$ is taken with the same value as in the first example, and $\delta g = g_0/2$ characterizes the inhomogeneous nonlinearity with $R=5\ell$ being fixed. The data generation is similar to the first example, except that now we adopt a basis of trigonometric polynomial, i.e., 
\begin{equation}
    \xi_n(x) = \frac{1}{R}\sin\left[\frac{n \pi}{2R}(x+R)\right],
    \label{xiSin}
\end{equation}
with $n_c=15$.
The trigonometric basis ensures that the complex field $\psi(x)$ vanishes at both boundaries of the box $x=\pm R$, which is reasonable for a boxed potential with hard walls. We train the CNO based on the training dataset $\mathcal{D}_\text{train}$. The hyper-parameters of the CNO are the same as those of the first example.

In Fig.~\ref{fig:box}(a), we show the variations of $\mathcal{L}$ as the training processes, which indicates that the training converges at about $400$ epochs. We set the upper energy bound to $E_\text{max} = 4$ for the soliton search, and as a consequence, the lowest 7 solitons are obtained and displayed in Fig.~\ref{fig:box}(b1)-(b3). Again, the soliton states obtained by Newton's method are also shown in Fig.~\ref{fig:box}(b) as a reference. It can be observed that, although the potential energy is flat inside the box, the amplitude of the field exhibits strong modulations due to the inhomogeneous nonlinearity. Accordingly, in Fig.~\ref{fig:box}(c), we show the overlap $F'$ [Eq.~(\ref{Fp})] and the energy error $\Delta E$ [Eq.~(\ref{DE})] based on the solitons obtained by the CNO and Newton's method, where an averaged $F'\ge 98\%$ and $\Delta E\le 0.13$ are indicated. Also in Fig.~\ref{fig:box}(c), $\Delta E$ shows an overall decreasing trend as $n$ increases, which is similar to the tendency in Fig.~\ref{fig:x2}(c) of the first example.

{\section{CNO-based Stability Analysis}\label{sec:NoiseAnalysis}}
{The trained CNO has learned the mapping between the initial and final states, which allows it to play the role of traditional PDE solvers in the numerical analysis of solitons. Furthermore, since the forward propagation of the CNO is quite faster than traditional PDE solvers as mentioned before, the numerics can be performed with higher efficiency. In this section, we demonstrate the CNO-based stability analysis on both examples above and compare its running time with that of traditional PDE solvers.}

Conventionally, one adds some small perturbations into the solitary waves, i.e., 
\begin{equation}
\psi(x,0) = \phi(x) + \epsilon\Delta \phi(x),
\label{DeltaPhi}
\end{equation}
and then evaluates the responses of the system after a period of evolution, which is known as the linear stability analysis. If the perturbation is not amplified, then the soliton is said to be linearly stable. Otherwise, the soliton is said to be linearly unstable. Solitary waves can also be unstable in a nonlinear fashion, which means the instability arises from the nonlinearity as $\epsilon$ is large enough to exceed the linear response regime. To fully understand the stability of a soliton, both linear and nonlinear effects should be considered.

\begin{figure}[!ht]
	\includegraphics[width=0.48\textwidth]{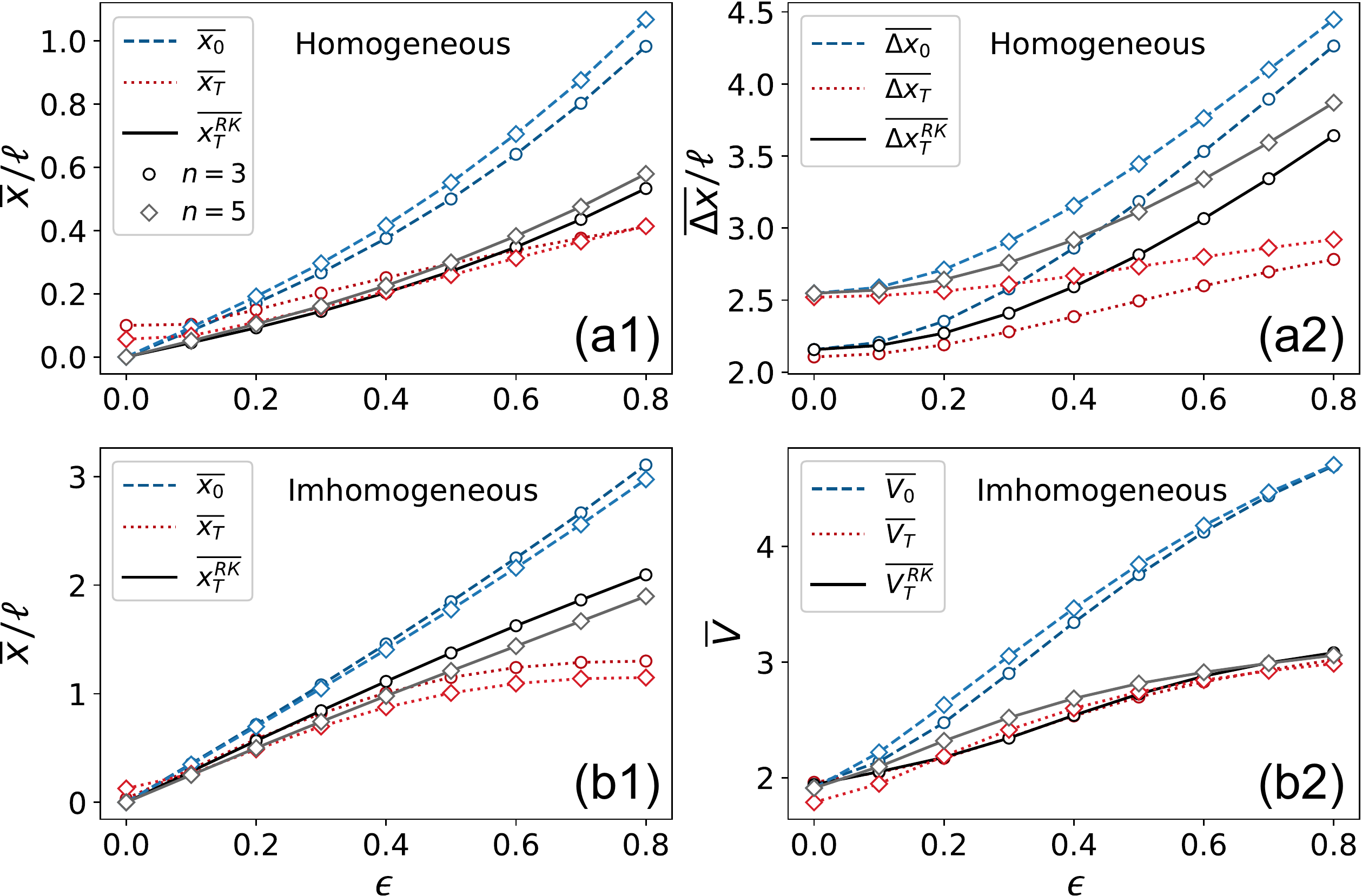}
	\caption{Stability analysis. (a) and (b) correspond to the cases of the homogeneous and inhomogeneous GP equations. (a1) and (b1) show the center of mass $\overline{x}$ as a function of $\epsilon$, where the dashed lines denote $\overline{x}$ of the initial state $\psi(x,0)$; the dotted lines and the solid lines indicate $\overline{x}$ of the final state $\psi(x,T)$ using the CNO and the RK method, respectively. (a2) shows the variation of the envelop width $\overline{\Delta x}$. (b2) shows the variation of the density visibility $\overline{V}$. In each panel, circles and diamonds correspond to the results of the third and the fifth solitons, respectively.}
	\label{fig:stability}
\end{figure}

We carry out the stability analysis of the two examples by feeding the trained CNO with perturbed initial states $\psi(x,0)$ [Eq.~(\ref{DeltaPhi})], where $\Delta \phi(x) = \sum_n c_n \xi_n$ is randomly generated using the basis Eqs.~(\ref{xiHG}) and (\ref{xiSin}), and $\epsilon$ denotes the overall amplitude of perturbation. {After an evolution time of $T$, the response to perturbations is embedded in the output of CNO, i.e., $\psi_\text{pred}(x,T)$. We then examine several observables to quantify the responses.} For the first example, we look at the center-of-mass of the BEC $\overline{x_{t=\{0,T\}}}=\overline{|\langle x(t) \rangle|}$ and its envelop width $\overline{\Delta x_{t=\{0,T\}}}=\overline{\sqrt{\langle x^2(t) \rangle - \langle x(t) \rangle^2}}$, where
\begin{equation}
\begin{aligned}
    \langle x(t) \rangle &= {\frac{1}{N} \int dx x|\psi_\text{pred}(x,t)|^2}, \\
    \langle x^2(t) \rangle &= {\frac{1}{N} \int dx x^2|\psi_\text{pred}(x,t)|^2}.
\end{aligned}
\end{equation}
The additional overscore on the R.H.S. of the $\overline{x_{t}}$ and $\overline{\Delta x_{t}}$ means to take a further average on different perturbations $\Delta \phi(x)$ under a fixed $\epsilon$. For the second example, since the BEC always spreads diffusely inside the box with density modulations, we hence replace $\overline{\Delta x_{t}}$ by the visibility of density, i.e.,{
\begin{equation}
\overline{V_t} = \overline{\frac{|\psi_\text{pred}(x,t)|^2_\text{max} - |\psi_\text{pred}(x,t)|^2_\text{min}}{\rho_\text{avg}}},
\label{vt}
\end{equation}
with $|\psi_\text{pred}(x,t)|^2_{\text{max,min}}$ being the maximal (minimal) density and $\rho_\text{avg}=N/2R$ denoting the averaged density.}

Figs.~\ref{fig:stability}(a) and (b) display the stability calculations of the two examples using the CNO, respectively. To be more specific, the left panels (a1) and (b1) show $\overline{x_{0}}$ (dashed lines) and $\overline{x_{T}}$ (dotted lines) of the 3rd soliton $\phi_3$ (circles) and the 5th soliton $\phi_5$ (diamonds) as functions of the perturbation strength $\epsilon$; the right panel (a2) shows the envelop width $\overline{\Delta x_{0}}$ (dashed lines) and $\overline{\Delta x_{T}}$ (dotted lines) of the first example (homogeneous GP equation); the panel (b2) shows the density visibility $\overline{V_0}$ (dashed lines) and $\overline{V_T}$  (dotted lines) of the second example (inhomogeneous GP equation). In each panel, the black solid line ($\overline{x_{T}^\text{RK}}$, $\overline{\Delta x_{T}^\text{RK}}$, or $\overline{V_{T}^\text{RK}}$) is plotted using the traditional { RK finite-difference method which uses the 4th-order RK formula to deal with the time and the finite difference for the spatial discretization. The RK results serve as the ground truth of the evolution.} One can clearly observe in Fig.~\ref{fig:stability} that the CNO predictions are in qualitative agreement with those obtained by the RK method, especially in the regime of small $\epsilon$. Furthermore, for either $\overline{x}$ or $\overline{\Delta x}$, the response at time $t=T$ is always smaller than that at the initial time $t=0$, which indicates the solitons found out by the CNO are stable. Particularly in the linear response regime $\epsilon\ll 1$, the responses are roughly linear in $\epsilon$, whereas, for $\epsilon$ far away from the linear response regime, both examples exhibit apparent nonlinear behaviors.

{
Finally, let us compare the running time of CNO and traditional PDE solvers. Two traditional algorithms, explicit RK (same as before) and explicit Euler's methods, are considered here. Practically, we evolve the perturbed $\phi_3$ state (with fixed $\epsilon=0.1$) from $t=0$ to $t=T=1/\omega$, repeat this process 100 times, and then take the average of the total running time. Fig.~\ref{fig:time} (a) and (b) respectively show the running time of various methods in the homogeneous and inhomogeneous examples, where the horizontal axis is a time-space ratio \cite{Scherer2013}
\begin{equation}
r=\frac{\delta t}{\delta x^2},
\end{equation}
with $\delta t$ and $\delta x$ being the minimal steps in time and space, respectively. $r$ can be seen as the generalized Courant number \cite{Courant1928}, being closely related to the stability of algorithms. Generally, explicit algorithms would become unstable as $r$ increases. We remind that the stability here refers to the stability condition of a PDE solver, rather than the stability of soliton solutions mentioned above. Algorithms that do not satisfy the stability condition would generate exponentially large numerical errors. Our calculation shows that, in Fig.~\ref{fig:time}, the solid line (RK) and the dotted line (Euler) connected areas are the stable areas, and the unconnected areas on the right are unstable. In our calculations, $\delta x=0.04\ell$ is fixed. To reduce $r$ means to reduce $\delta t$, and hence the total running time behaves $\propto1/\delta t$, i.e., a straight line with a slope of $-1$ in the log-log plots. 

It is clearly shown in Fig.~\ref{fig:time} that, for a given stable $\delta t$ (or $r$), the computational time of CNO is significantly less time than that of RK and Euler's methods, and is also $\delta t$-independent. This is consistent with our complexity analysis in Sec.~\ref{CNO}. Specifically, we chose $l=4$, which is thus a $\delta t$-independent constant being much smaller than $t_s = T/\delta t$. On the other hand, for certain stable $\delta t$, Euler always consumes less time than RK since Euler's method performs fewer calculations within one $\delta t$ iteration.}

\begin{figure}[t]
    \includegraphics[width=0.48\textwidth]{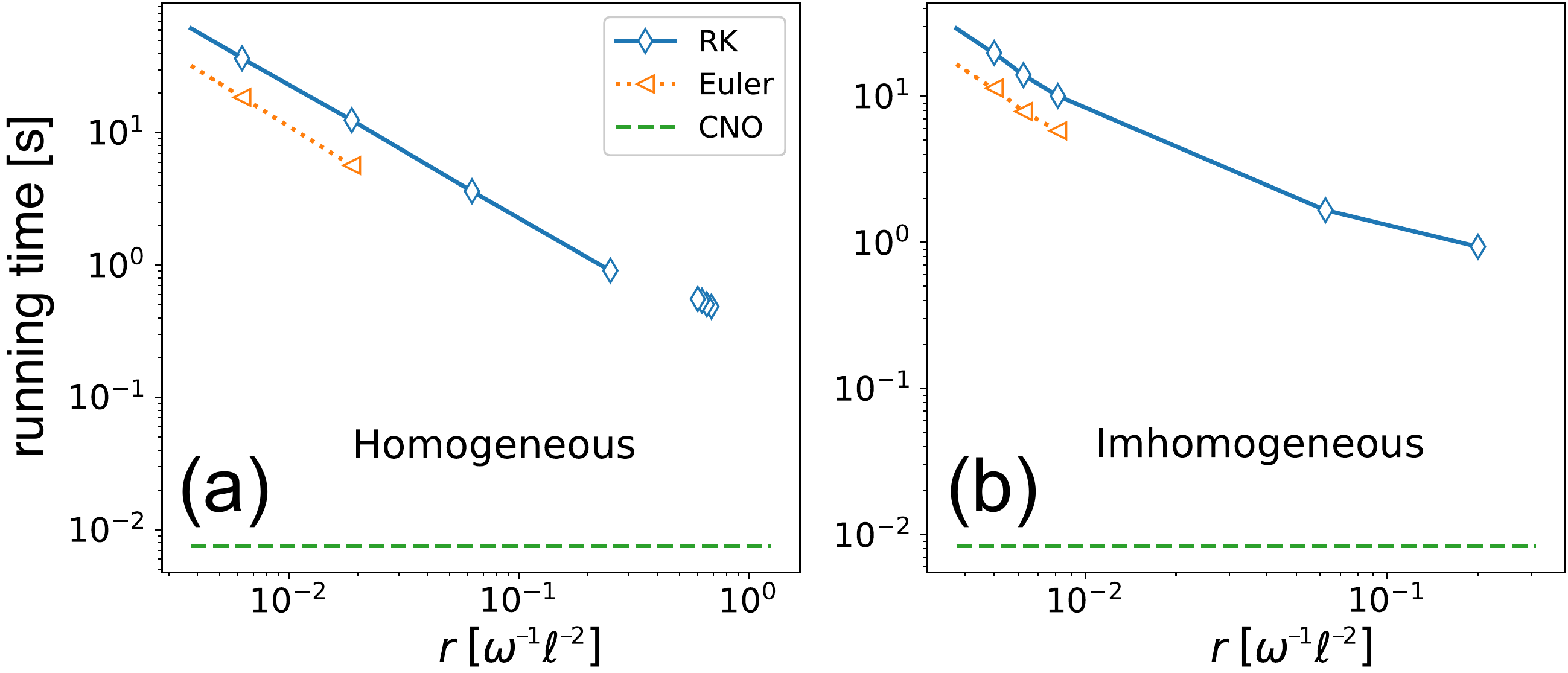}
	\caption{Running time as a function of the time-space ratio $r$ on the examples of homogeneous (a) and inhomogeneous (b) GP equations. In each panel, dashed lines, triangles, and diamonds indicate the running time of CNO, explicit Euler, and explicit RK methods, respectively. For the latter two methods, the connected dotted and solid lines label out the regions where algorithms are stable.}
	\label{fig:time}
\end{figure}

\section{Summary and Discussion}
\label{sec:Summary}
We proposed a data-based approach to search for the solitary solutions of dynamical PDEs. By extending the FNO architecture to the complex field, we developed an architecture called the CNO which can learn the complex mapping between the initial and final states. A combination of the CNO and an energy-restricted search algorithm allows us to identify solitary waves within a limited space of states. Our approach was demonstrated on one-dimensional GP equations with both homogeneous and inhomogeneous nonlinearities, and the resulting solitons exhibited a high overlap with those derived by Newton's method. We also showed the application of the CNO to the stability analysis of solitons. In terms of computational complexity, CNO-based calculations were more efficient than traditional PDE solvers. 

There are a number of follow-up questions. {Extending our algorithm to higher-dimensional systems is not as straightforward as one might expect. In the above 1D examples, we used a dataset of size $10^4$ to ensure that the FNO can capture the initial-final state mapping $G$. In higher dimensions, much more data points are needed, which places a heavy burden on data generation. As a result, randomly generating data in higher dimensions is no longer feasible. In addition, the training process for larger FNO models is also slow and time-consuming. One possible solution is to adopt semi-supervised learning algorithms, such as active learning \cite{Settles2012, Yao2020, Chen2020}, which can serve as an efficient data acquisition strategy to optimize the data generation and learning efficiency.} Furthermore, although symmetries are useful in reducing the search space, as illustrated in the examples, this does not mean that all symmetries are exploitable. In the case of the GP equation with spin-orbit coupling \cite{Lin2011, Li2012, Chen2017} or the Klein-Gordon equation with negative mass \cite{Goldstone1961, Peskin1995}, spontaneous symmetry breaking would lead to phase transitions, such that the solitary states exhibit lower symmetry than the Lagrangian. In such a case, the broken symmetry cannot be used as a constraint. In turn, this motivates us to think about how to use the CNO or some other machine-learning algorithms to identify phase transitions. We expect this work, as well as these questions, to prompt more studies in the fields of machine learning and many-body quantum physics.

\begin{acknowledgments}
This work was done when the first author was visiting Microsoft Research.
L.C. would like to thank Huan-bo Luo for the helpful discussion. L.C. acknowledges support from the National Natural Science Foundation of China (Grant Nos. 12174236 and 12147215) and from the fund for the Shanxi 1331 Project.
\end{acknowledgments}

\appendix

\section{Compliments for the Energy-Restricted Search}\label{app:ERS}
In Figs.~\ref{fig:x2}(c) and \ref{fig:box}(c), we observed the energy error $\Delta E$ exhibiting a decreasing tendency as $n$ increases. This phenomenon can be attributed to the fact that the energy constraint (the last term in Eq.~(\ref{equ:OP_New})) becomes increasingly significant as $E_{\psi(x,0)}$ approaches the bound $E_\text{max}$, rendering the searching algorithm less likely to be trapped in a local minimum, as has been mentioned in the paragraph above Eq.~(\ref{DE}). In the following, we provide a mathematical proof of this statement.
\begin{prop} \label{prop}
In the optimization problem (P): $\min_\psi \Delta F$ subject to $E_\psi \leq E_\text{max}$, for any $n$ ($E_{\phi_n} \leq E_\text{max}$), there is $E_{ \psi_n^*}$ (or $\Delta E_n$) monotonically increasing about $E_\text{max}$, where $\psi_n^*$ is the optimal solution to the problem (P).
\end{prop}

\begin{proof}
For any $n$, take $E_\text{max}^1$ (P1) and $E_\text{max}^2$ (P2), satisfying $E_{\phi_n}\leq E_\text{max}^2 < E_\text{max}^1$. Denote the feasible region of problems (P1) and (P2) being $s_1$ and $s_2$ respectively, then $s_2 \subset s_1$.

For the problem (P1), if $\psi_{n,1}^* \in s_2$, then $\psi_{n,1}^*=\psi_{n,2}^*$. In this case, $E_{\psi_{n,1}^*}=E_{\psi_{n,2}^*}$.
Else if $\psi_{n,1}^* \in s_1\backslash s_2$, and $\psi_{n,2}^* \in s_2$, then $E_{\psi_{n,2}^ *}\leq E_\text{max}^2 < E_{\psi_{n,1}^*}$.

In summary, when $E_{\phi_n}\leq E_\text{max}^2 < E_\text{max}^1$, there is $E_{\psi_{n,2}^*} \leq E_{\psi_{n ,1}^*}$ ($\Delta E_n^2 \leq \Delta E_n^1$). That is, $E_{ \psi_n^*}$ ($\Delta E_n$) monotonically increases with respect to $E_\text{max}$.
\end{proof}


\end{document}